\documentclass[12pt]{article}
\usepackage{fullpage}
\usepackage{amsmath,amssymb,amsthm}
\usepackage{geometry}
\usepackage{xcolor}
\usepackage[normalem]{ulem}

\newtheorem{theorem}{Theorem}[section]
\newtheorem{proposition}[theorem]{Proposition}

\newtheorem{corollary}[theorem]{Corollary}

\newtheorem{remark}[theorem]{Remark}

\title{The Hodograph Transform Between Thermodynamics and Relativity}
\author{Leonid Polterovich\thanks{Tel Aviv University, Israel}}
\date{\today}

\begin{document}

\maketitle

\begin{abstract}
In the contact-geometric approach to general relativity, the sky of an event—namely, the set of all incoming light rays—forms a Legendrian submanifold of the spherical cotangent bundle of a Cauchy hypersurface. When the hypersurface is chosen to be the Minkowski hyperboloid, a hyperbolic version of the hodograph transform identifies this bundle with a thermodynamic phase space. We consider a uniformly accelerating observer starting on the hyperboloid and study the evolution of its skies. We show that the associated generating functions, after a suitable rescaling, admit a natural interpretation as reduced free energies of equilibrium thermodynamic systems governed by the relativistic Doppler effect. From this data, we extract an effective temperature that is proportional to the acceleration, in agreement with the scaling of the Unruh effect, although the numerical constant differs from the Unruh value.
\end{abstract}

\section{Introduction}

It was recently understood that the causal structure on the space of Legendrian submanifolds—a deep phenomenon in contact geometry—appears both in general relativity  \cite{CN} and in thermodynamics \cite{EPR}. In this note we suggest a toy model which might be viewed as an argument in favor of a duality between these two subjects.

We view the unit future hyperboloid in the $(2+1)$-dimensional Minkowski space as a Cauchy hypersurface of the chronological future of the origin,
which is identified with the hyperbolic upper half-plane $\mathbb{H}$.
We consider an observer who starts at an event $A$ on the hyperboloid with constant proper acceleration $a$, and in proper time $t$ arrives at an event $B$.

The sky (i.e.\ the set of all incoming light rays) of the event $B$ is viewed as a Legendrian submanifold of the spherical cotangent bundle $S^*\mathbb{H}$ of the hyperbolic plane $\mathbb{H}$ .
Next, we apply a version of the hodograph transform $\mathcal{H}$ \cite{AP,Fe,CN} between the jet space $J^1S^1$ and the spherical cotangent bundle $S^*\mathbb{H}$. The latter can be viewed as a thermodynamic phase space, with the directions of light rays in Minkowski space serving as control (intensive) variables.

Our main result (see Theorem \ref{thm-main} below) provides an explicit expression for the generating function of $\mathcal{H}^{-1}(S)$, the preimage of the sky $S$ of the event $B$. Namely, it yields a function $g_t$ on $S^1$ whose $1$-jet $j^1 g_t$ coincides with $\mathcal{H}^{-1}(S)$.

The formulas for $g_t$ naturally involve quantities of thermodynamic flavor, most notably the energy (given by the Doppler factor, i.e.\ the ratio of photon energies measured by the observers at $B$ and $A$, see Section \ref{sec-Doppler}). In the thermodynamic interpretation of the phase space $J^1S^1$, Legendrian submanifolds represent equilibrium states, and their generating functions serve as thermodynamic potentials, such as free energy. We therefore interpret $g_t$ as a reduced free energy.

After an appropriate Lorentzian time rescaling—motivated by the Unruh effect
\footnote{See Wikipedia contributors
 ``Unruh effect.'' Wikipedia, The Free Encyclopedia.},
according to which uniformly accelerating observers exhibit thermal behavior—the resulting formulas allow us to extract an effective temperature. The corresponding  temperature scales as $\sim a$, in agreement with the Unruh temperature scaling,
\begin{equation}
\label{eq-unruh}
T_{unruh}=\frac{a}{2\pi},
\end{equation}
although the numerical constant we obtain differs from the Unruh value.

Guided by this scaling, we construct a thermodynamic system whose reduced free energy reproduces the leading behavior of $g_t$ (see  Theorem \ref{thm-free-eta} in Section \ref{sec-free}).
This system is related to a classical rotor (see Section \ref{subsec-rot}).

The paper is organized as follows. In Section \ref{sec-Doppler} we discuss the Busemann
function of the hyperbolic plane and its relation to the photon energy measured by an observer, as well as to the Doppler effect. In Section \ref{sec-hodog} we introduce the hyperbolic hodograph transform, whose properties are proved in Section \ref{sec-Appendix-A}.
In Section \ref{sec-main} we present our model of an accelerating observer and state
our main result, Theorem \ref{thm-main},  which provides an expression for the generating function of the sky of the observer under the hodograph transformation.
Section \ref{sec-free} relates this generating function to the free energy of a thermodynamic system.

\medskip
\noindent
{\bf Convention on units:}  Unless otherwise stated, all quantities in this paper are dimensionless. In particular, we work in normalized geometric units associated with the unit hyperboloid, so that proper time, acceleration, photon energy, and temperature are all understood after division by the appropriate reference scales. We explain how to restore physical units in Section \ref{subsec-units}.

\section{Busemann Function as a Logarithmic Energy Profile}\label{sec-Doppler}

Let $K \subset \mathbb{R}^{1,2}$ be the unit future hyperboloid
\[
K = \{ x \in \mathbb{R}^{1,2} \mid
\langle x,x\rangle = -1, \ x_0>0 \},
\]
equipped with the hyperbolic metric induced by the Minkowski form
\[
\langle u,v\rangle
=
- u_0 v_0 + u_1 v_1 + u_2 v_2.
\]
Denote by $\partial_\infty K$ the ideal boundary of $K$ which we identify
with the set of future light-like rays starting at the origin. Note that
$\partial_\infty K$ is diffeomorphic to the circle.

Fix a ray $q \in \partial_\infty K$. Let $\gamma(t)$ be any geodesic ray in $K$ asymptotic to $q$,
parametrized by hyperbolic arclength. Write $d$ for the hyperbolic distance.
The \emph{Busemann function} associated with $q$ is defined by
\[
b_q(x)
=
\lim_{t\to\infty}
\big( d(x,\gamma(t)) - t \big), \;\; x \in K\;.
\]
The Busemann function associated with $q$ is defined only up to an additive constant; shifting the origin of the geodesic ray asymptotic to $q$ changes $b_q$ by a constant. Fixing a reference observer  removes this ambiguity.

Let us calculate the Busemann function.
Fix a reference point $A \in K$. Take any ray $q \in \partial_\infty K$. Choose the (unique) null vector $\hat q$ spanning $q \in \partial_\infty K$, so that
\begin{equation}\label{eq-norm}
\langle A, \hat q\rangle = -1.
\end{equation}
One readily checks that the curve
$$\gamma(t):= \frac{e^t}{2}\hat{q} + \frac{e^{-t}}{2}(2A-\hat{q}), t \geq 0$$
is a geodesic ray with the natural parameterization which lies on $K$, satisfies
$\gamma(0)=A$,  and which is
asymptotic to $q$ as $t \to +\infty$.

 In the hyperboloid model, the hyperbolic distance satisfies
\[
\cosh(d(x,y))
=
- \langle x,y\rangle.
\]
Thus
$$d(x,\gamma(t)) = \operatorname{arccosh}\left(-\frac{e^t}{2}\langle x, \hat{q}\rangle - \frac{e^{-t}}{2}\langle x, 2A-\hat{q}\rangle\right)\;.$$
Using the asymptotic expansion
\[
\operatorname{arccosh}(z)=\log(2z)+o(1)
\qquad\text{as }z\to\infty,
\]
we obtain
\[
d(x,\gamma(t))
=
t+\log(-\langle x,\hat q\rangle)+o(1).
\]
Thus, the Busemann function is given by
\begin{equation}\label{eq-bus-1}
b_q(x)
=
\log(-\langle x,\hat q\rangle).
\end{equation}
In what follows we fix once and for all a reference observer $A\in K$.
Equality \eqref{eq-norm} removes the scaling ambiguity for the choice of a light vector in the given ray.

Following \cite{Pereira}, Section 3.3, put $\varepsilon_x(q):= -\langle  x, \hat q\rangle$.
When $x\in K$, it is a unit future--directed timelike vector and hence can
be interpreted as an observer. In this case, $\varepsilon_x(q)$ is the energy
(proportional to the frequency) of the photon with null momentum $\hat q$
measured by the observer represented by $x$.
By construction, $\varepsilon_{A}(q)=1$ for all $q$.
By \eqref{eq-bus-1}
\begin{equation}\label{eq-bus-2}
b_q(x)=\log \frac{\varepsilon_x(q)}{\varepsilon_{A}(q)}= \log \varepsilon_x(q)\;.
\end{equation}
The dimensionless ratio ${\varepsilon_x(q)}/{\varepsilon_{A}(q)}$  reflects the relativistic Doppler effect.

Let us mention that the quantities appearing in our formulas — namely, the photon energy ratios (Doppler factors) — coincide with the redshift factors, which, as shown in \cite{CN-red}, are given by the ratios of contact forms on the space of light rays associated to different Cauchy hypersurfaces.
\section{Space of light rays and hodograph transform}\label{sec-hodog}

The space of future light geodesics in Minkowski space can be
identified with the unit cosphere bundle $S^*K$, which we identify with $S^* \mathbb{H}$.
This is a contact manifold. Let us describe the contactomorphism
\[
\mathcal{H} : J^1 S^1 \longrightarrow S^*\mathbb{H},
\]
which we call the \emph{hyperbolic hodograph transform}.
It identifies the unit cosphere bundle of the hyperbolic plane
with the 1--jet space of the boundary circle.

The 1--jet space
\[
J^1 S^1 = T^*S^1 \times \mathbb{R}
\]
carries the standard contact form
\[
\alpha = dz - p\,dq,
\]
where $(q,p,z)$ are coordinates with
$q \in S^1$, $p \in T^*_q S^1$, and $z \in \mathbb{R}$.
Its Reeb vector field is $R = \partial_z$.

The unit cosphere bundle
\[
S^*\mathbb{H} = \{ (u,\xi) \in T^*\mathbb{H} : |\xi|_{g_{\mathbb{H}}}=1 \}
\]
carries the Liouville contact form $\lambda$.
Its Reeb flow is the geodesic flow of the hyperbolic metric.

To define the hodograph transform, we view the Busemann function as a function
$b_q : \mathbb{H} \to \mathbb{R}$.
Given $(q,p,z) \in J^1 S^1$, we set
\[
\mathcal{H}(q,p,z) = (u, db_q(u)),
\]
where $u \in \mathbb{H}$ is determined by the conditions
\begin{align*}
b_q(u) &= z, \\
\frac{\partial b_q(u)}{\partial q} &= p.
\end{align*}

\begin{proposition}\label{prop-hod}
The map $\mathcal{H}$ is a contactomorphism, and furthermore
\[
\mathcal{H}^*\lambda
=
\alpha\;,
\]
\end{proposition}

The hodograph transformation is well known, see  \cite{Ar,AP,CN}.
The hyperbolic version which we are using is a minor modification of the
construction in \cite{Fe}. For reader's convenience, we present a proof of
Proposition \ref{prop-hod} in Section \ref{sec-Appendix-A}, while in Section \ref{sec-Appendix-B} we revise the Euclidean hodograph transformation from the viewpoint of
the Busemann functions.

\medskip

Time $s$ geodesic flow shifts the Busemann function by a constant,
\[
b_q \mapsto b_q + s,
\]
so in $J^1S^1$ this flow corresponds to the vertical translation along the $z$-axis.

\medskip

Let us now discuss the action of the inverse hodograph map on skies. For $u \in \mathbb{H}$,
the fiber $S^*_{u}\mathbb{H}$
is mapped by $\mathcal{H}^{-1}$ to the Legendrian graph $j^1f$ of $f(q)= b_q(u)$.
Hence the skies lying on the Cauchy hypersurface correspond to the Legendrian graphs of Busemann functions, and their images under the time $s$ maps of the geodesic flow correspond to the Legendrian graphs of $b_q(u) + s$. Combining this consideration with
\eqref{eq-bus-2}, we get the following result.

\begin{proposition}\label{prop-hod-1} Let $S_{u,s}$ be the image of $S^*_{u}\mathbb{H}$
under time $s$ map of the geodesic flow. Then $\mathcal{H}^{-1}(S_{u,s})$
is generated by
\begin{equation}
\label{eq-genf-prem}
g_s(q) = s+\log \varepsilon_u(q)\;.
\end{equation}
\end{proposition}

\section{Case study: accelerated motion}\label{sec-main}

\subsection{Set up}\label{subsec-setup}

We view the unit future hyperboloid
\[
K=\{x\in\mathbb{R}^{1,2}\mid \langle x,x\rangle=-1,\ x_0>0\},
\]
as a Cauchy hypersurface of the chronological future of the origin, i.e.,
of the cone
\[
\{(t,x,y)\in \mathbb{R}^{1,2}\;:\; t>0,\; x^2+y^2<t^2\}.
\]
Let
\[
A=(1,0,0)\in K
\]
be the reference observer used for the normalization of photon energies.

Let $n\in T_AK$ be a spatial unit vector:
\[
\langle n,n\rangle=1, \qquad \langle A,n\rangle=0 .
\]

Consider the uniformly accelerated worldline starting at $A$ with
proper acceleration $a>0$ and initial four--velocity $A$. It has the
explicit form
\[
\gamma(t)=
\Bigl(1+\frac{\sinh(at)}{a}\Bigr)A
+
\frac{\cosh(at)-1}{a}\,n ,
\]
where $t$ is the proper time. Denote
\[
B=\gamma(t).
\]

The point $B$ lies in the interior of the future light cone.
To describe the sky of $B$ on the fixed hypersurface $K$, we project
$B$ radially to $K$ by normalizing its Minkowski length:
\[
B'=\frac{B}{\sqrt{-\langle B,B\rangle}} .
\]
Then $\langle B',B'\rangle=-1$, so $B'\in K$.

Geometrically, $B'$ is the unique point of $K$ lying on the ray
$\mathbb{R}_{>0}B$.

The past light cone of $B$ intersects $K$ in the hyperbolic circle
\[
\{u\in K\mid d(u,B')=\rho(t)\},
\]
where
\begin{equation}\label{eq:rho}
\rho(t):=\frac{1}{2}\log\bigl(-\langle B,B\rangle\bigr),
\end{equation}
and $d$ is the hyperbolic distance on $K$.

Indeed, a point $u\in K$ lies on the past light cone of $B$ if and only if
\[
\langle u-B,u-B\rangle=0.
\]
Write
\[
\lambda:=\sqrt{-\langle B,B\rangle}>0,
\qquad
B=\lambda B'.
\]
Since $u,B'\in K$, we have
\[
\langle u,u\rangle=\langle B',B'\rangle=-1.
\]
Therefore
\[
0=\langle u-B,u-B\rangle
=\langle u,u\rangle-2\langle u,B\rangle+\langle B,B\rangle
=-1-2\lambda\langle u,B'\rangle-\lambda^2.
\]
Hence
\[
-\langle u,B'\rangle=\frac{\lambda+\lambda^{-1}}{2}.
\]
On the other hand, in the hyperboloid model the hyperbolic distance satisfies
\[
\cosh d(u,B')=-\langle u,B'\rangle.
\]
Thus
\[
\cosh d(u,B')
=
\frac{\lambda+\lambda^{-1}}{2}
=
\cosh(\log \lambda).
\]
Since $d(u,B')\ge 0$ and $\lambda>1$, we get
\[
d(u,B')=\log \lambda
=\frac12\log\bigl(-\langle B,B\rangle\bigr).
\]
This proves that the intersection of the past light cone of $B$ with $K$ is precisely
\[
\{u\in K\mid d(u,B')=\rho(t)\},
\]
where
$$
\rho(t)=\frac12\log\bigl(-\langle B,B\rangle\bigr).
$$
Thus the sky of $B$ on $K$ is the unit conormal lift
$
S_{B',-\rho(t)},
$
where
\begin{equation}
\rho(t)
=
\frac{1}{2}\log\!\left(
1
+
\frac{2}{a}\sinh(at)
+
\frac{2}{a^2}\bigl(\cosh(at)-1\bigr)
\right).
\end{equation}

\subsection{The measured photon energy.}

Let
\begin{equation}\label{eq-fourvel}
u_t:=\dot\gamma(t)=\cosh(at)\,A+\sinh(at)n
\end{equation}
be the four--velocity of the accelerated observer. Consider a photon with
the null momentum $\hat q$.  The corresponding
photon energy, as measured by the observer, is
\begin{equation}\label{eq-phen}
E_t(q):=-\langle u_t,\hat q\rangle
=
\cosh(at)-\sinh(at)\langle n,\hat q\rangle .
\end{equation}
For small $t$,
\begin{equation}\label{eq-Et-small}
E_t(q) = 1- a \langle n, \hat q\rangle t + O(t^2)\;.
\end{equation}

\subsection{Main theorem}

\begin{theorem}\label{thm-main}
\begin{itemize}
\item[{(i)}] The generating function $g_t$ of
$\mathcal{H}^{-1}(S_{B',-\rho(t)})$ is given by
\begin{equation}\label{eq-g_t-doppler}
g_t(q)=\log \varepsilon_{B'}(q)-\rho(t).
\end{equation}
\item[{(ii)}] The generating function $g_t(q)$ admits the following explicit expression in terms of $E_t(q)$:
 \begin{equation}
\begin{aligned}
g_t(q)
&=
\log\!\left(
1+\frac{1}{a}\tanh\!\left(\frac{at}{2}\right)(1+E_t(q))
\right) \\
&\quad
-
\log\!\left(
1+\frac{2}{a}\sinh(at)
+\frac{2}{a^2}(\cosh(at)-1)
\right).
\end{aligned}
\end{equation}
  \item[{(iii)}] We have the following small $t$ expansion:
\begin{equation}\label{eq-genf-main}
g_t(q)
=
-\frac{3}{2}t
+
\frac{1}{2}\,t^2
+
\frac{1}{2}\,tE_t(q)
+
O(t^3).
\end{equation}
\end{itemize}
%Its small-$t$ expansion has the form
%\begin{equation}\label{eq-genf-main}
%g_t(q)
%=
%-t+\Bigl(\frac{1}{2}-\pi T\,\langle n,\hat q\rangle\Bigr)t^2
%+O(t^3),
%\end{equation}
%where $T$ is the Unruh temperature.
\end{theorem}

\medskip

The right-hand sides of \eqref{eq-g_t-doppler} and \eqref{eq-genf-main}
are closely related to the relativistic Doppler effect for accelerating observers.
In Section \ref{sec-free} we interpret $g_t$ in terms of free energy.

\begin{proof}
\medskip\noindent{\bf ({i}):} Formula \eqref{eq-g_t-doppler} follows from Proposition~\ref{prop-hod-1}
with $s=-\rho(t)$.

\medskip\noindent{\bf ({ii}):}
By part {\rm(i)},
\[
g_t(q)=\log \varepsilon_{B'}(q)-\rho(t).
\]
Recall that
\[
\varepsilon_{B'}(q)
=
\frac{\alpha(t)-\beta(t)\langle n,\hat q\rangle}
{\sqrt{D(t)}},
\qquad
\rho(t)=\frac12\log D(t),
\]
where
\[
\alpha(t)=1+\frac{\sinh(at)}{a},
\qquad
\beta(t)=\frac{\cosh(at)-1}{a},
\]
and
\[
D(t)=
1+\frac{2}{a}\sinh(at)+\frac{2}{a^2}\bigl(\cosh(at)-1\bigr).
\]
Hence
\[
g_t(q)
=
\log\!\bigl(\alpha(t)-\beta(t)\langle n,\hat q\rangle\bigr)
-
\log D(t).
\]

Now
\[
E_t(q)=\cosh(at)-\sinh(at)\langle n,\hat q\rangle,
\]
so
\[
\langle n,\hat q\rangle
=
\frac{\cosh(at)-E_t(q)}{\sinh(at)}.
\]
Therefore
\[
\alpha(t)-\beta(t)\langle n,\hat q\rangle
=
1+\frac{\sinh(at)}{a}
-
\frac{\cosh(at)-1}{a}\cdot
\frac{\cosh(at)-E_t(q)}{\sinh(at)}.
\]
Combining the terms gives
\[
\alpha(t)-\beta(t)\langle n,\hat q\rangle
=
1+\frac{\cosh(at)-1}{a\,\sinh(at)}\bigl(1+E_t(q)\bigr).
\]
Using
\[
\frac{\cosh(at)-1}{\sinh(at)}
=
\tanh\!\left(\frac{at}{2}\right),
\]
we obtain
\[
\alpha(t)-\beta(t)\langle n,\hat q\rangle
=
1+\frac{1}{a}\tanh\!\left(\frac{at}{2}\right)\bigl(1+E_t(q)\bigr).
\]
Substituting this into the formula for $g_t(q)$ yields
\[
\begin{aligned}
g_t(q)
&=
\log\!\left(
1+\frac{1}{a}\tanh\!\left(\frac{at}{2}\right)(1+E_t(q))
\right)
\\
&\quad
-
\log\!\left(
1+\frac{2}{a}\sinh(at)
+\frac{2}{a^2}(\cosh(at)-1)
\right),
\end{aligned}
\]
as claimed.

\medskip\noindent{\bf ({iii}):}
Using the formula proved in {\rm(ii)},
\[
\begin{aligned}
g_t(q)
&=
\log\!\left(
1+\frac{1}{a}\tanh\!\left(\frac{at}{2}\right)(1+E_t(q))
\right)
\\
&\quad
-
\log\!\left(
1+\frac{2}{a}\sinh(at)
+\frac{2}{a^2}(\cosh(at)-1)
\right).
\end{aligned}
\]

Since
\[
\frac1a\tanh\!\left(\frac{at}{2}\right)=\frac t2+O(t^3),
\qquad
E_t(q)=1+O(t),
\]
we get
\[
\frac1a\tanh\!\left(\frac{at}{2}\right)(1+E_t(q))
=
\frac t2(1+E_t(q))+O(t^3).
\]
Hence
\[
\log\!\left(
1+\frac1a\tanh\!\left(\frac{at}{2}\right)(1+E_t(q))
\right)
=
\frac t2(1+E_t(q))
-\frac{t^2}{8}(1+E_t(q))^2
+O(t^3).
\]
Since \(E_t(q)=1+O(t)\), we have
\[
(1+E_t(q))^2=4+O(t),
\]
and therefore
\[
\log\!\left(
1+\frac1a\tanh\!\left(\frac{at}{2}\right)(1+E_t(q))
\right)
=
\frac t2+\frac t2E_t(q)-\frac12 t^2+O(t^3).
\]

Also,
\[
\sinh(at)=at+O(t^3),
\qquad
\cosh(at)-1=\frac{a^2t^2}{2}+O(t^4),
\]
so
\[
1+\frac{2}{a}\sinh(at)+\frac{2}{a^2}(\cosh(at)-1)
=
1+2t+t^2+O(t^3),
\]
and thus
\[
\log\!\left(
1+\frac{2}{a}\sinh(at)+\frac{2}{a^2}(\cosh(at)-1)
\right)
=
2t-t^2+O(t^3).
\]

Subtracting,
\[
g_t(q)
=
\left(\frac t2+\frac t2E_t(q)-\frac12 t^2\right)
-
(2t-t^2)
+O(t^3),
\]
which gives
\[
g_t(q)
=
-\frac32 t+\frac12 t^2+\frac12 tE_t(q)+O(t^3).
\]
This proves \eqref{eq-genf-main}.
\end{proof}

\medskip
\noindent
\begin{remark}{\rm
The use of the hyperboloid $K$ as a Cauchy hypersurface is essential for
the thermodynamic interpretation.

Indeed, the Busemann function $b_q(x)$ depends on the spacetime event $x$,
while the photon energy $\varepsilon_u(q)=-\langle u,\hat q\rangle$
is defined with respect to an observer, i.e.\ a future unit timelike vector $u$.

On the hyperboloid $K$ these notions coincide, since every point
$x\in K$ can be naturally interpreted as an observer with velocity $x$.
This yields the identity
\[
b_q(x)=\log \varepsilon_x(q),
\]
which underlies the free energy interpretation.

In contrast, if one chooses a Euclidean Cauchy hypersurface such as
$\{t=t_0\}$, there is no canonical identification between points and
observer velocities. As a result, the Busemann function and the photon
energy depend on different geometric objects, and no natural energy-type
term arises on the thermodynamic side.}
\end{remark}

\subsection{Comparison with the Unruh temperature}\label{subsec-unruh}

\paragraph{Derivation of the effective inverse temperature.}
We start from \eqref{eq-genf-main}:
\[
g_t(q)= -\frac{3}{2}\,t+\frac{1}{2}\,t^2+\frac{1}{2}\,t\,E_t(q)+O(t^3),
\]
where
\[
E_t(q)=\cosh(at)-\sinh(at)\langle n,\hat q\rangle.
\]

Introduce the new time parameter (called a boost parameter or the Rindler time)\footnote{The boost parameter $\eta$ is dimensionless (it is a rapidity). In physical units one has $\eta = a\tau/c$, where $c$ is the speed of light.}
\[
\eta:=a t,
\qquad t=\frac{\eta}{a},
\]
and define
\begin{equation}\label{eq-deff}
\widetilde g_\eta(q):=g_{\eta/a}(q),
\qquad
f_\eta(q):=\frac{\widetilde g_\eta(q)}{\eta}
=\frac{g_{\eta/a}(q)}{\eta}.
\end{equation}

Substituting $t=\eta/a$ into the expansion of $g_t$ yields
\[
\widetilde g_\eta(q)
=
-\frac{3\eta}{2a}
+\frac{\eta^2}{2a^2}
+\frac{\eta}{2a}E_{\eta/a}(q)
+O(\eta^3).
\]
Dividing by $\eta$, we obtain
\begin{equation}\label{eq-f-forent}
f_\eta(q)
=
-\frac{3}{2a}
+\frac{\eta}{2a^2}
+\frac{1}{2a}E_{\eta/a}(q)
+O(\eta^2).
\end{equation}

Hence, up to a $q$-independent term,
\begin{equation}\label{eq-fe-affine}
f_\eta(q)
=
C(\eta)+\frac{1}{2a}E_{\eta/a}(q)+O(\eta^2),
\end{equation}
where $C(\eta)$ does not depend on $q$.

Recall that the reduced free energy (or the Massieu function, with the opposite sign)
of a thermodynamic system equals $\beta U - S$, where $S$ is the entropy, and $U$ the internal energy. At this point we make an ``ansatz" and interpret $f_\eta$ as the reduced free energy, and $E=E_{\eta/a}(q)$ as the internal energy.

We consider $f_\eta$ as a function of $E$
(cf. Theorem \ref{thm-main})(i)) and define the effective inverse temperature by
\[
\beta_{\mathrm{eff}}
:=
\left.\frac{\partial f_\eta}{\partial E}\right|_{\eta=0}.
\]
Using \eqref{eq-fe-affine}, we obtain
\begin{equation}\label{eq-beta-eff}
\beta_{\mathrm{eff}}
=
\frac{1}{2a}.
\end{equation}
 The assumption that $f_\eta$ is a kind of free energy (or, more generally, a thermodynamic potential) is motivated by the fact that we view the jet space $J^1S^1$ as a thermodynamic phase space, and the physically meaningful Legendrians as equilibrium submanifolds of a thermodynamic system.

\paragraph{Relation to the boost parameter and the Unruh effect.}
The parameter $\eta=at$ admits a natural interpretation as the \emph{Lorentz boost parameter} along the worldline of a uniformly accelerating observer. Indeed, the velocity of the observer is given by
\[
\dot\gamma(t)=\cosh(at)\,A+\sinh(at)\,n,
\]
so that $at$ is precisely the rapidity (hyperbolic angle) of the corresponding Lorentz transformation.

%In quantum field theory, the Unruh effect is formulated in terms of evolution with respect to the boost parameter $\eta$.\footnote{More precisely, the Minkowski vacuum, when restricted to a Rindler wedge, is a thermal (KMS) state with respect to the generator of Lorentz boosts. The associated imaginary-time periodicity is $2\pi$ in the variable $\eta$ \cite{W}.}
%We treat this as a black box and use only the fact that the natural time variable in the Unruh effect is the boost parameter $\eta$, rather than the proper time $t$.

In quantum field theory, the boost parameter $\eta$ plays an important role in the derivation of the Unruh effect.\footnote{More precisely, the Minkowski vacuum, when restricted to a Rindler wedge, is a thermal (KMS) state with respect to the generator of Lorentz boosts. The associated imaginary-time periodicity is $2\pi$ in the variable $\eta$ \cite{W}.}
We take this as a black box and use $\eta$ in the subsequent considerations.

Rewriting the generating function $g_t$ in terms of $\eta$ and introducing the normalized quantity $f_\eta$ therefore allows for a direct comparison with the thermodynamic structure underlying the Unruh effect. Formula \eqref{eq-beta-eff} shows that the effective inverse temperature has the same $a^{-1}$ scaling as the Unruh temperature \eqref{eq-unruh}.

\paragraph{Heuristic connection to the Unruh effect.}
The discussion below is based on the simplified derivation of the Unruh effect \cite{AM}.

An accelerating observer may be viewed as probing the ambient field by receiving waves arriving from all lightlike directions $q\in S^1$. Each such direction corresponds to a mode of propagation, and the observer measures its energy through the relativistic Doppler factor
\[
E_t(q)=-\langle \dot\gamma(t),\hat q\rangle.
\]
For an inertial observer these energies would be constant in time, but along an accelerated trajectory they acquire a nontrivial time dependence. In particular, the frequencies of the incoming waves become time-dependent, so that each mode is perceived not as a pure harmonic oscillation, but as a signal with time-dependent frequency (a so-called ``chirped'' signal) whose instantaneous frequency evolves along the trajectory.

The total signal seen by the observer is therefore a superposition over all directions $q$ of such oscillations with time-varying frequencies. A natural way to analyze this signal is to pass to its spectral decomposition in the observer’s proper time. In heuristic derivations of the Unruh effect, this spectral analysis yields a distribution of frequencies of the form
\[
N(\Omega)\propto \frac{1}{e^{2\pi \Omega c/a}-1}.
\]
Here $N(\Omega)$ denotes the spectral density of the signal observed along the trajectory, that is, the squared modulus of its Fourier transform, representing the distribution of frequencies in proper time. This has exactly the Planck form
\[
\frac{1}{e^{\hbar\Omega/(kT)}-1},
\]
and the temperature is then identified by matching the exponents. This derivation, following \cite{AM}, yields the Unruh temperature \eqref{eq-unruh} (in units where $\hbar = k = c = 1$).

In our model, instead of performing a Fourier analysis in time, we encode the same Doppler energy profile geometrically. The function
\[
g_t(q)=\log \varepsilon_{B'}(q)-\rho(t)
\]
records, for each direction $q$, the logarithmic energy of the corresponding mode. In this sense, $g_t$ may be viewed as a generating function that packages the directional dependence of the observed frequencies. Its small-time expansion involves the energies $E_t(q)$, and, after an appropriate rescaling, it admits an interpretation as a reduced free energy. The derivative of this quantity with respect to the energy yields an effective inverse temperature, which scales linearly with the acceleration, in agreement with the Unruh scaling.

\subsection{Restoration of physical units}\label{subsec-units}

So far, we have worked with the unit hyperboloid and therefore with
dimensionless quantities. In particular, the time variable $t$, the acceleration
parameter $a$, and the function $E_t(q)$ are all dimensionless.

To restore physical units, fix a length scale $R>0$ and replace the unit hyperboloid by
\[
K_R=\{x\in\mathbb{R}^{1,2}:\langle x,x\rangle=-R^2,\ x_0>0\}.
\]
Let $\tau$ denote the physical proper time and $\alpha$ the physical proper acceleration.
The dimensionless variables used throughout the paper are then given by
\[
t=\frac{c\tau}{R},
\qquad
a=\frac{\alpha R}{c^2},
\qquad
\eta=at=\frac{\alpha\tau}{c}.
\]

With this identification, all formulas of the previous sections remain valid without
modification. In particular, the function $E_t(q)$ depends on $t$ only through the
combination $at=\eta$, so that
\[
E_t(q)=E_{\eta/a}(q)=\cosh\eta-\sinh\eta\langle n,\hat q\rangle.
\]

We define the physical energy by
\[
\mathcal E_\eta(q):=\gamma E_{\eta/a}(q)
=\gamma\bigl(\cosh\eta-\sinh\eta\langle n,\hat q\rangle\bigr),
\]
where $\gamma>0$ is a constant with units of energy.

With this convention, formula \eqref{eq-f-forent} takes the form
\begin{equation}\label{eq-f-forent-units}
f_\eta(q)
=
-\frac{3}{2a}
+\frac{\eta}{2a^2}
+\frac{1}{2a\gamma}\,\mathcal E_\eta(q)
+O(\eta^2).
\end{equation}

Since $f_\eta$ is dimensionless, the coefficient of $\mathcal E_\eta$ has units of
inverse energy and can be interpreted as an effective inverse temperature:
\[
\beta_{\mathrm{eff}}=\frac{1}{2a\gamma}
=
\frac{c^2}{2\alpha R\,\gamma}.
\]
Accordingly,
\[
T_{\mathrm{eff}}=\frac{2\alpha R\,\gamma}{c^2}.
\]

In particular, the temperature is proportional to the acceleration, in agreement
with the Unruh scaling, although the numerical constant depends on the choice of
the energy scale $\gamma$.

\section{A statistical interpretation as reduced free energy}\label{sec-free}

In this section we introduce randomness in the direction of the spatial
acceleration vector and interpret the resulting generating function in
terms of reduced free energy.

\subsection{Randomization of the acceleration direction}

We model the spatial acceleration direction by a random unit vector
$n\in S^1$. Fix a preferred direction $n_0\in S^1$ and assume that $n$
is distributed according to the von Mises distribution
\[
d\mu_\kappa(n)
=
\frac{1}{2\pi I_0(\kappa)}
\exp\big(\kappa \langle n,n_0\rangle\big)\,d\phi,
\qquad \kappa\ge0,
\]
where $I_0$ is the modified Bessel function of the first kind, defined by
\[
I_0(\kappa)
=
\frac{1}{2\pi}\int_0^{2\pi} e^{\kappa\cos\theta}\,d\theta.
\]
We also recall that
\[
I_1(\kappa)
=
\frac{1}{2\pi}\int_0^{2\pi} e^{\kappa\cos\theta}\cos\theta\,d\theta.
\]
The parameter $\kappa$ controls the concentration: for $\kappa=0$ the
distribution is uniform, while for $\kappa\to\infty$ it concentrates
near $n_0$. Introduce the boost parameter
\[
\eta:=at, \qquad t=\frac{\eta}{a}.
\]
With this notation the four--velocity of the accelerating
observer from our model is given by
$u_{\eta/a} = \cosh\eta\,A+\sinh\eta\,n$,  see \eqref{eq-fourvel}.

Define the {\it Hamiltonian} as the photon energy\footnote{To keep the units consistent, the expression on the right-hand side of \eqref{eq-Heta} should be multiplied by a constant, say $\gamma$, with units of energy. For simplicity, we tacitly assume $\gamma=1$.}
\begin{equation}\label{eq-Heta}
H_\eta(n,q)
:=
\varepsilon_{u_{\eta/a}}(q) =
\cosh\eta-\sinh\eta\langle n,\hat q\rangle.
\end{equation}

For a future use, recall that $\langle A,\hat q\rangle=-1$.
Write $A=(1,0,0)$ and $\hat q=(\alpha,\hat q_{\mathrm{sp}})$, where
$\hat q_{\mathrm{sp}}$ denotes the spatial component of $\hat q$.
It follows that $\alpha=1$ and $|\hat q_{\mathrm{sp}}|^2=1$.

Define the partition function
\[
Z_\eta(q)
:=
\int_{S^1} e^{-\beta H_\eta(n,q)}\,d\mu_\kappa(n),
\]
where $\beta$ is the inverse temperature.
Substituting \eqref{eq-Heta}, we obtain
\[
Z_\eta(q)
=
\frac{e^{-\beta\cosh\eta}}{2\pi I_0(\kappa)}
\int_0^{2\pi}
\exp\Big(
\kappa\langle n,n_0\rangle
+
\beta\sinh\eta\,\langle n,\hat q\rangle
\Big)\,d\phi.
\]

A direct computation yields
\begin{equation}\label{eq-Zeta}
Z_\eta(q)
=
e^{-\beta\cosh\eta}
\frac{I_0(R_\eta(q))}{I_0(\kappa)},
\end{equation}
where
\begin{equation}\label{eq-Reta}
R_\eta(q)
=
\sqrt{
\kappa^2
+
\beta^2\sinh^2\eta
+
2\kappa\beta\sinh\eta\langle n_0,\hat q\rangle
}.
\end{equation}

Indeed,
\[
\kappa\langle n,n_0\rangle+\beta\sinh \eta\langle n,\hat q\rangle
=
\langle n,v_\eta(q)\rangle,
\]
where
\[
v_\eta(q):=\kappa n_0+\beta\sinh \eta \,\hat q_{\mathrm{sp}}
\]
and \(\hat q_{\mathrm{sp}}\), as above, denotes the spatial component of \(\hat q\).
Hence
\[
|v_\eta(q)|
=
\sqrt{
\kappa^2+\beta^2\sinh^2\eta+2\kappa\beta\sinh\eta\langle n_0,\hat q\rangle
}
=
R_\eta(q).
\]

Writing \(v_\eta(q)=R_\eta(q)\,m_\eta(q)\) with \(m_\eta(q)\in S^1\), we obtain
\[
\kappa\langle n,n_0\rangle+\beta\sinh\eta\langle n,\hat q\rangle
=
R_\eta(q)\langle n,m_\eta(q)\rangle.
\]
Therefore
\[
\int_{S^1}
\exp\Big(
\kappa\langle n,n_0\rangle+\beta\sinh\eta\langle n,\hat q\rangle
\Big)\,d\phi
=
\int_{S^1} e^{R_\eta(q)\langle n,m_\eta(q)\rangle}\,d\phi.
\]

Since the measure \(d\phi\) is invariant under rotations of \(S^1\), the
integral depends only on \(R_\eta(q)\) and not on the direction \(m_\eta(q)\).
Thus we can rotate coordinates so that \(m_\eta(q)=(1,0)\), and obtain
\[
\int_{S^1} e^{R_\eta(q)\langle n,m_\eta(q)\rangle}\,d\phi
=
\int_0^{2\pi} e^{R_\eta(q)\cos\phi}\,d\phi
=
2\pi I_0(R_\eta(q)).
\]

Substituting this into the definition of \(Z_\eta(q)\) yields
\[
Z_\eta(q)
=
e^{-\beta\cosh\eta}\frac{I_0(R_\eta(q))}{I_0(\kappa)}.
\]
Equation \eqref{eq-Zeta} follows.

Define the {\it reduced free energy} by
\[
F_\eta(q):=-\log Z_\eta(q).
\]
Using \eqref{eq-Zeta}, we obtain
\begin{equation}\label{eq-Feta}
F_\eta(q)
=
\beta\cosh\eta
-
\log I_0(R_\eta(q))
+
\log I_0(\kappa).
\end{equation}

\begin{theorem}\label{thm-free-eta}
Assume that the spatial acceleration direction is distributed according
to the von Mises law with parameter $\kappa>0$. Then, as $\eta\to0$,
\[
F_\eta(q)
=
\beta
-
\beta\,\frac{I_1(\kappa)}{I_0(\kappa)}
\,\eta\,\langle n_0,\hat q\rangle
+
O(\eta^2).
\]
\end{theorem}

\begin{proof}
Set
\[
s_\eta:=\beta\sinh\eta =
\beta\eta+O(\eta^3).
\]
Then
\[
R_\eta(q)
=
\sqrt{\kappa^2+s_\eta^2+2\kappa s_\eta\langle n_0,\hat q\rangle}
=
\kappa+s_\eta\langle n_0,\hat q\rangle+O(\eta^2).
\]
Expanding $\log I_0$ at $\kappa$, we get (by using that $I'_0=I_1$) that
\[
\log I_0(R_\eta(q))
=
\log I_0(\kappa)
+
\frac{I_1(\kappa)}{I_0(\kappa)}\,s_\eta\langle n_0,\hat q\rangle
+
O(\eta^2).
\]

Substituting into \eqref{eq-Feta} yields
\begin{equation}\label{eq-fetaq}
F_\eta(q)
=
\beta
-
\beta\frac{I_1(\kappa)}{I_0(\kappa)}
\,\eta\,\langle n_0,\hat q\rangle
+
O(\eta^2).
\end{equation}
This completes the proof.
\end{proof}

Recall that
\[
E_{\eta/a}(q)
=
\cosh\eta-\sinh\eta\langle n_0,\hat q\rangle = 1- \eta \langle n_0,\hat q\rangle + O(\eta^2).
\]

\begin{corollary}\label{cor-concentration-eta}
Put $\beta=1/(2a)$ (see \eqref{eq-beta-eff}).
As $\kappa\to+\infty$, $\eta\to0$,
\begin{equation}\label{eq-Eeta}
F_\eta(q)
=
\beta\,E_{\eta/a}(q)
+
O\big(\kappa^{-1}\eta+\eta^2\big).
\end{equation}
Thus, if $\kappa^{-1}=O(\eta)$,
the rescaled generating function $f_\eta(q):= \eta^{-1}g_{\eta/a}(q)$
satisfies
$$f_\eta(q) = C(\eta) + F_\eta(q) + O(\eta^2)\;.$$
\end{corollary}

\begin{proof}
The corollary immediately follows from \eqref{eq-Et-small}, \eqref{eq-fetaq}, \eqref{eq-fe-affine}, \eqref{eq-deff}, and the expansion
$$\frac{I_1(\kappa)}{I_0(\kappa)} = 1 + O(1/\kappa), \; \kappa \to +\infty\;.$$
\end{proof}

\subsection{Comparison with the classical rotor}\label{subsec-rot}

Consider a classical rotor whose configuration space is the unit circle
$S^1$ (cf. \cite{DS}). Let $u=(\cos\theta,\sin\theta)$ be a fixed unit vector representing the
direction of an external field. The Hamiltonian of the rotor is
\[
H_{\rm rot}(\phi)=-h\,n(\phi) \cdot u,
\qquad
n(\phi)=(\cos\phi,\sin\phi).
\]
Here dot stands for the scalar product in $\mathbb{R}^2$.
The constant $h$, with units of energy, will be specified later.

At the inverse temperature $\beta>0$ the partition function is
\[
Z_{\rm rot}(h,\beta)
=
\int_{S^1}
e^{-\beta H_{\rm rot}(\phi)}\frac{d\phi}{2\pi}
=
\int_0^{2\pi}
e^{(\beta h)\cos(\phi-\theta)}\frac{d\phi}{2\pi}
=
I_0\!\left(\beta h\right),
\]
where $I_0$ is the modified Bessel function. The corresponding reduced free energy is
\[
F_{\rm rot}(h,\beta)=-\log Z_{\rm rot}(h,\beta).
\]

We consider the small--field regime with $T$ fixed. Using the standard expansion
\[
I_0(x)=1+\frac{x^2}{4}+O(x^4),
\]
we obtain
\[
F_{\rm rot}(h,\beta)
=
-\frac{\beta^2 h^2}{4} +O(h^4).
\]

Recall (see Theorem \ref{thm-free-eta}) that the photon free energy is given by
\[
F_\eta(q)
=
\beta
-
\beta\,\frac{I_1(\kappa)}{I_0(\kappa)}\,\eta\,\langle n_0,\hat q\rangle
+
O(\eta^2).
\]

Assume that
\[
\langle n_0,\hat q\rangle>0\;.
\]
Choose the field strength
\[
h
=
2\sqrt{\beta^{-1}
\frac{I_1(\kappa)}{I_0(\kappa)} \eta
\langle n_0,\hat q\rangle
}.
\]
Substituting this into the rotor free energy gives
\begin{equation}
F_\eta(q)
=
\beta + F_{\rm rot}(h,\beta)
+
O(\eta^2).
\end{equation}

Hence, after an appropriate scaling of the field strength, the free energy
of the classical rotor captures the leading directional dependence of the
photon free energy, up to a $q$--independent term. Note that the field strength
$h$ depends on $q$; thus we are not comparing with a single fixed rotor, but rather with a family of rotor models whose field is tuned to match the directional dependence.

\section{Hodograph transform-miscellanea}

\subsection{Hyperbolic hodograph-proofs} \label{sec-Appendix-A}

\medskip
\noindent
{\bf Proof of Proposition \ref{prop-hod}:}
We use coordinates $(u,v)$ on $S^*\mathbb{H}$, with $u \in \mathbb{H}$
and $v$ being a unit covector at $u$. Recall that $\lambda = v\,du$.
Write $b_q(u)$ as $b(q,u)$. The hodograph map has the form
$\mathcal{H}(q,p,z) = (U,V)$ with $V= \frac{\partial b}{\partial u}(q,U)$.
Since $z= b(q,U)$ and $p=  \frac{\partial b}{\partial q}(q,U)$, we get
\[
dz= p\,dq + V\,dU,
\]
so that $V\,dU = dz-p\,dq$. This proves the proposition.
$\Box$

\medskip

Next, let us prove that $\mathcal{H}$ is one-to-one and onto. To this end we work
in the unit disc model of the hyperbolic plane.
Let $\mathbb D=\{u\in\mathbb C:\ |u|<1\}$ be the Poincar\'e disk with metric
\[
g_u=\frac{4}{(1-|u|^2)^2}\,|du|^2.
\]
We identify $K$ with the Poincaré disk via the standard radial projection
from the origin onto the plane $x_0=1$.
Note that under the identification of $K$ with $\mathbb D$, the ideal boundary $\partial_\infty K$ corresponds to the boundary $S^1=\partial \mathbb D$.
Without loss of generality, the marked point
in the definition of the Busemann function is the origin.
For $q\in S^1$  the  Busemann function is given by
\[
b_q(u):=\log\frac{1-|u|^2}{|u-q|^2}, \qquad u\in\mathbb D.
\]

\begin{proposition} \label{prop:hodograph-diffeo}
Let $S^*\mathbb D=\{(u,\xi):\ \|\xi\|_{g^*_u}=1\}$ be the unit cosphere bundle.
Define the map
\[
\Theta:S^*\mathbb D\longrightarrow J^1S^1, \qquad
\Theta(u,\xi)=(q,p,z),
\]
as follows: $q=q(u,\xi)\in S^1$ is the unique point such that
\begin{equation}\label{eq:xi-is-db}
\xi=d_u b_q,
\end{equation}
and then
\begin{equation}\label{eq:def-zp}
z=b_q(u),\qquad p=\frac{\partial}{\partial\phi}\,b_{e^{i\phi}}(u)\Big|_{e^{i\phi}=q},
\end{equation}
where $\phi$ is an angular coordinate on $S^1$ (hence $p$ is the covector component
in the canonical trivialization $T^*S^1\cong S^1\times\mathbb R$).
Then $\Theta$ is a diffeomorphism whose inverse is the hodograph map $\mathcal H$.
\end{proposition}

\begin{proof}
\textbf{Step 1: $\Theta$ is well-defined.}
We first record two standard facts about Busemann functions on $(\mathbb D,g)$:

\smallskip
\emph{(a) Unit speed.} For every $q\in S^1$, the Busemann function $b_q$ satisfies the eikonal equation
\begin{equation}\label{eq:eikonal}
\|d_u b_q\|_{g^*_u}=1,\qquad \forall u\in\mathbb D.
\end{equation}
(Equivalently, $\|\nabla_g b_q\|_g=1$.)

\smallskip
\emph{(b) Uniqueness of $q$ from $(u,\xi)$.}
Fix $u\in\mathbb D$. The map
\[
S^1\ni q\longmapsto d_u b_q\in S^*_u\mathbb D
\]
is a diffeomorphism. Geometrically, $d_u b_q$ is the unit covector pointing in the
direction of the (unique) geodesic ray from $u$ to $q$; distinct $q$ give distinct
forward directions, and every forward direction determines a unique ideal endpoint $q$.
Hence for each $(u,\xi)\in S^*\mathbb D$ there exists a unique $q$ satisfying \eqref{eq:xi-is-db},
and then \eqref{eq:def-zp} defines $(p,z)$ uniquely. Thus $\Theta$ is well-defined.

\medskip
\textbf{Step 2: An explicit inverse map $J^1S^1\to S^*\mathbb D$.}
Write $q=e^{i\phi}\in S^1$ and set $a:=e^z>0$.
We shall reconstruct $u$ from $(q,p,z)$ and then set $\xi=d_u b_q$.

By rotational symmetry, it suffices to solve the reconstruction for $q=1$, and then rotate back.
Let $R_q:\mathbb D\to\mathbb D$ be the rotation $R_q(w)=qw$.
Since $b_{q}(u)=b_1(\bar q u)$ (because $|u-q|=|\bar q u-1|$ and $|u|=|\bar q u|$),
if we set $w=\bar q u$ then
\[
z=b_q(u)=b_1(w),
\qquad
p=\frac{\partial}{\partial\phi}b_{e^{i\phi}}(u)\Big|_{e^{i\phi}=q}
=
\frac{\partial}{\partial\phi}b_{e^{i\phi}}(qw)\Big|_{\phi=\arg q}
=
\frac{\partial}{\partial\phi}b_{e^{i\phi}}(w)\Big|_{\phi=0}.
\]
Thus, without loss of generality, we may assume that $q=1$ and solve the equation
\[
\Theta(w,\xi)=(1,p,z).
\]
Once $w$ is determined, the covector is uniquely given by $\xi = d_w b_1$.
Write $w=x+iy$.
Set
\[
d:=|w-1|^2=(x-1)^2+y^2,\qquad r^2:=|w|^2=x^2+y^2.
\]
From the definitions,
\begin{equation}\label{eq:system-ap}
a=e^{z}=\frac{1-r^2}{d},
\qquad
p=\frac{\partial}{\partial\phi}b_{e^{i\phi}}(w)\Big|_{\phi=0}
=\frac{2y}{d}.
\end{equation}
(The last identity is a direct differentiation of $b_{e^{i\phi}}(w)=\log\frac{1-r^2}{|w-e^{i\phi}|^2}$
at $\phi=0$.)

From \eqref{eq:system-ap} we express $y=\frac{pd}{2}$ and $r^2=1-ad$.
Also $d=|w-1|^2=r^2-2x+1$, hence
\begin{equation}\label{eq:x-from-d}
x=1-\frac{(1+a)d}{2}.
\end{equation}
Now $r^2=x^2+y^2$ together with $r^2=1-ad$, $y=\frac{pd}{2}$, and \eqref{eq:x-from-d}
gives a single positive solution
\begin{equation}\label{eq:d-solution}
d=\frac{4}{(1+a)^2+p^2},
\end{equation}
and then
\begin{equation}\label{eq:w-solution}
x=\frac{a^2+p^2-1}{(1+a)^2+p^2},
\qquad
y=\frac{2p}{(1+a)^2+p^2}.
\end{equation}
Equivalently,
\begin{equation}\label{eq:w-closed}
w(a,p)=\frac{a^2+p^2-1}{(1+a)^2+p^2}
+i\,\frac{2p}{(1+a)^2+p^2}.
\end{equation}
Finally, define for general $q\in S^1$:
\begin{equation}\label{eq:inverse-u}
u(q,p,z):=q\,w(e^z,p).
\end{equation}
This yields a smooth map
\[
\mathcal G:J^1S^1\to \mathbb D,\qquad \mathcal G(q,p,z)=u(q,p,z).
\]
Moreover, from \eqref{eq:w-closed} one checks
\begin{equation}\label{eq:w-in-disc}
|w(a,p)|^2=\frac{(a-1)^2+p^2}{(a+1)^2+p^2}<1,
\end{equation}
so indeed $u(q,p,z)\in\mathbb D$ for all $(q,p,z)\in J^1S^1$.

Now set
\begin{equation}\label{eq:inverse-xi}
\mathcal H(q,p,z):=\bigl(u(q,p,z),\, d_{u(q,p,z)} b_q\bigr)\ \in\ S^*\mathbb D.
\end{equation}
The membership in $S^*\mathbb D$ follows from \eqref{eq:eikonal}.

\medskip
\textbf{Step 3: $\Theta$ and $\mathcal H$ are inverse to each other.}
We show $\Theta\circ \mathcal H=\mathrm{id}$ and $\mathcal H\circ \Theta=\mathrm{id}$.

\smallskip
\emph{(i) $\Theta\circ \mathcal H=\mathrm{id}_{J^1S^1}$.}
Take $(q,p,z)\in J^1S^1$ and set $(u,\xi)=\mathcal H(q,p,z)$.
By construction $\xi=d_u b_q$, so the first component of $\Theta(u,\xi)$ is exactly $q$.
It remains to check that the resulting $(p',z')$ coincide with $(p,z)$.
But $u$ was obtained by solving the system \eqref{eq:system-ap} with $a=e^z$ and $p$,
hence for $q=1$ we have $b_1(w)=z$ and $\partial_\phi b_{e^{i\phi}}(w)|_{\phi=0}=p$.
Rotating back (using $w=\bar q u$) yields
\[
b_q(u)=z,\qquad \frac{\partial}{\partial\phi}b_{e^{i\phi}}(u)\Big|_{e^{i\phi}=q}=p.
\]
Therefore $\Theta(u,\xi)=(q,p,z)$.

\smallskip
\emph{(ii) $\mathcal H\circ \Theta=\mathrm{id}_{S^*\mathbb D}$.}
Take $(u,\xi)\in S^*\mathbb D$ and let $(q,p,z)=\Theta(u,\xi)$.
By definition, $\xi=d_u b_q$. Applying $\mathcal H$ produces
\[
\mathcal H(q,p,z)=\bigl(u(q,p,z),\, d_{u(q,p,z)} b_q\bigr).
\]
From part (i) we know that $\Theta(\mathcal H(q,p,z))=(q,p,z)$, so in particular the $u$--component
is uniquely determined by $(q,p,z)$. Since $(u,\xi)$ and $\mathcal H(q,p,z)$ have the same
$(q,p,z)$ under $\Theta$, their $u$--components coincide: $u(q,p,z)=u$.
Consequently $d_{u(q,p,z)}b_q=d_u b_q=\xi$, and hence $\mathcal H\circ\Theta=\mathrm{id}$.

Thus $\Theta$ is bijective with inverse $\mathcal H$.
\end{proof}

\subsection{Euclidean hodograph transform and Busemann functions} \label{sec-Appendix-B}

The classical Euclidean hodograph transform admits a natural interpretation
in terms of Busemann functions. Let $q\in S^1$ be a direction in
$\mathbb{R}^2$ and consider the ray $\gamma_q(t)=tq$, $t\ge0$. The associated
Busemann function is defined by
\[
b_q(x)=\lim_{t\to\infty}\big(d(x,\gamma_q(t))-t\big).
\]
A direct computation shows that in Euclidean space
\[
b_q(x)=-\langle x,q\rangle .
\]

Thus $b_q$ is a linear function whose gradient equals $-q$.
Using these functions one can define a map
\[
(x,q)\longmapsto j^1_q(b_q(x))
\]
from $\mathbb{R}^n\times S^1$ to the $1$--jet bundle $J^1S^1$.
Writing
\[
z=b_q(x), \qquad p=\partial_q b_q(x),
\]
this map takes the form
\[
(x,q)\longmapsto (q,p,z).
\]

One readily checks that this map coincides with the classical Euclidean
hodograph transform. Indeed, writing $q=(\cos\varphi,\sin\varphi)$ and
$x=(x_1,x_2)$, we have
\[
b_q(x)=-(x_1\cos\varphi+x_2\sin\varphi).
\]
Hence
\[
z=b_q(x)=-(x_1\cos\varphi+x_2\sin\varphi),
\]
and
\[
p=\partial_\varphi b_q(x)
= x_1\sin\varphi-x_2\cos\varphi .
\]
This, up to a sign convention, coincides with formula (6.2) in \cite{CN}.

This interpretation shows that the hodograph transform arises from the
family of Busemann functions associated with directions $q\in S^1$.
In particular, the Euclidean hodograph transform should be viewed as the
flat counterpart of the hyperbolic hodograph transform used above.
Indeed, replacing the Euclidean distance by the hyperbolic distance in
the definition of $b_q$ produces the hyperbolic Busemann functions, and
the same jet construction yields the hyperbolic hodograph map.

\section{Conclusion}

We propose a duality between the Lorentzian geometry of light rays and thermodynamics, based on the hyperbolic hodograph transform. The hodograph transform is a contactomorphism between the spherical cotangent bundle of the Minkowski hyperboloid in $\mathbb{R}^{1,2}$ and the jet space of the circle. The hyperboloid is a Cauchy hypersurface of the chronological future of the origin in Minkowski spacetime, while the jet space is viewed as a thermodynamic phase space. The intensive variables (points of the circle) correspond to future light directions in $\mathbb{R}^{1,2}$, and the thermodynamic potential (the $\mathbb{R}$-coordinate on $J^1S^1=\mathbb{R}\times T^*S^1$) corresponds to a reduced free energy.

We follow the trajectory of an observer starting on $K$ and undergoing uniform acceleration
with proper acceleration $a$. We track the evolution of its skies and consider their
(suitably rescaled) images under the hodograph map. The corresponding Legendrians are
given by generating functions, which we interpret as reduced free energies.
This yields a duality, some facets of which are summarized in the following table.

\bigskip

\begin{tabular}{p{0.48\textwidth} p{0.48\textwidth}}
\textbf{Thermodynamics} & \textbf{Relativity} \\ \hline
\textbf{1.} Legendrian submanifold generated by the reduced free energy
(thermodynamic equilibrium)
& Legendrian sky of an event \\
\textbf{2.} Internal energy $E$
& Photon energy measured in direction $q$
by two observers -- Doppler effect (Busemann function)  \\
\textbf{3.} Temperature $2a$  & Unruh temperature $a/2\pi$ \\
\textbf{4.} Entropy $S$ & Dimensionless boost parameter $at$ (Rindler time) \\
\textbf{5.} Reeb chord between Legendrians in \textbf{1} (ultrafast process)
& Light ray between events  \\
\textbf{6.} Quasistatic thermodynamic process & Deformation of skies along causal paths\\
\end{tabular}

\bigskip

Some comments are in order. The first two lines summarize the main constructions of this
paper, see \eqref{eq-fe-affine} and \eqref{eq-phen}, respectively. The third line is explained
in Section \ref{subsec-unruh}.

Line 4 deserves a special discussion. Recall that in equilibrium thermodynamics,
if $f=\beta F$ denotes the reduced free energy, then the entropy is given by
\[
s=\beta \frac{\partial f}{\partial \beta}-f.
\]

Using the unit-consistent form \eqref{eq-f-forent-units},
\[
f_\eta(q)
=
-\frac{3}{2a}
+\frac{\eta}{2a^2}
+\frac{1}{2a\gamma}\,\mathcal E_\eta(q)
+O(\eta^2),
\]
we rewrite it in thermodynamic form as
\[
f_\eta(q)
=
-3\gamma\beta
+
2\eta(\gamma\beta)^2
+
\beta\,\mathcal E_\eta(q)
+
O(\eta^2),
\]
where $\eta$ is viewed as a parameter and
\[
\beta=\frac{1}{2a\gamma}
\]
is the effective inverse temperature.

A direct computation yields
\[
s=2\eta(\gamma\beta)^2.
\]
Substituting $\beta=(2a\gamma)^{-1}$, we obtain
\[
s=\frac{\eta}{2a^2}.
\]

Thus the entropy is proportional to the boost parameter $\eta$ (often called the Rindler time).
This is natural, since entropy and time encode the causal structures
in thermodynamics and relativity, respectively (see e.g.\ \cite{HP}).

Finally, lines 5 and 6 allude to the classification of thermodynamic processes within
the contact-topological framework proposed in \cite{EPR}.

As a final remark, all our results extend to higher dimensions.
The corresponding hodograph map is a contactomorphism between the spherical cotangent bundle
$S^*\mathbb{H}^n$ of the $n$-dimensional hyperbolic space and the jet bundle of the
$(n-1)$-dimensional sphere, $J^1S^{n-1}$. The physical model and its consequences remain intact.

\medskip
\noindent{\bf Acknowledgment.} This paper was written during the 2026 program
``Contact geometry, general relativity and thermodynamics'' at the Simons Center for Geometry and Physics at Stony Brook. I thank the Simons Center for the excellent research atmosphere.
I am especially grateful to Shin-itiro Goto for numerous helpful comments and suggestions.
I also thank Alessandro Bravetti, Olaf M\"{u}ller, and Miguel S\'{a}nchez Caja for useful discussions and valuable feedback on the manuscript, as well as Stefan Nemirovski for providing references on the hyperbolic hodograph map.


\begin{thebibliography}{99}
\bibitem{Pereira} Pereira RF. The kinematic space of special relativity and its hyperbolic geometry (Doctoral dissertation, Universidade de São Paulo).
\end{thebibliography}
\end{document}